\documentclass[sigconf]{acmart}

\usepackage{icse-style}
\usepackage{xspace}
\usepackage{enumitem}
\usepackage{todonotes}
\usepackage[htt]{hyphenat}
\usepackage{flushend}

\theoremstyle{remark}
\newtheorem{RQ}{RQ}

\newcommand{\doneByFD}[1]{\color{magenta}{#1}\color{black}}

\newcommand{\extslice}{extended slice\xspace}

\newcommand{\conextslice}{cut\xspace}
\newcommand{\conextslices}{cuts\xspace}
\newcommand{\etal}{\emph{et al.}}
\newcommand{\pdepa}{\textsf{pdepa}\xspace}

\begin{document}


\copyrightyear{2020}
\acmYear{2020}
\setcopyright{acmlicensed}
\acmConference[ICSE '20]{42nd International Conference on Software Engineering}{May 23--29, 2020}{Seoul, Republic of Korea}
\acmBooktitle{42nd International Conference on Software Engineering (ICSE '20), May 23--29, 2020, Seoul, Republic of Korea}
\acmPrice{15.00}
\acmDOI{10.1145/3377811.3380372}
\acmISBN{978-1-4503-7121-6/20/05}

\title{Lazy Product Discovery in Huge Configuration Spaces}

\author[M. Lienhardt]{Michael Lienhardt}
\affiliation{
	\institution{ONERA \textminus\ The French Aerospace Lab}
	\country{France}
}
\email{michael.lienhardt@onera.fr}
\author[F. Damiani]{Ferruccio Damiani}
\affiliation{
	\institution{University of Turin}
	\country{Italy}
}
\email{ferruccio.damiani@unito.it}
\author[E.~B. Johnsen]{Einar Broch Johnsen}
\affiliation{
	\institution{University of Oslo}
	\country{Norway}
}
\email{einarj@ifi.uio.no}
\author[J. Mauro]{Jacopo Mauro}
\affiliation{
	\institution{University of Southern Denmark}
	\country{Denmark}
}
\email{mauro@sdu.dk}

\begin{abstract}
  Highly-configurable software systems can have thousands of
  interdependent configuration options across different subsystems. In
  the resulting configuration space, discovering a valid product
  configuration for some selected options can be complex and error
  prone.  The configuration space can be organized using
  a feature model, fragmented into smaller interdependent feature
  models reflecting the configuration options of each subsystem.

  We propose a method for lazy product discovery in large fragmented
  feature models with interdependent features.  We formalize the
  method and prove its soundness and completeness.  The evaluation
  explores an industrial-size configuration space. The results show
  that lazy product discovery has significant performance benefits
  compared to standard product discovery, which in contrast to our
  method requires all fragments to be composed to analyze the feature
  model.  Furthermore, the method succeeds when more efficient,
  heuristics-based engines fail to find a valid configuration.
\end{abstract}


\maketitle

\section{Introduction}\label{sec:introduction}

Highly-configurable software systems can have thousands of
interdependent configuration options across different subsystems. In
the resulting configuration space, different software variants can be
obtained by selecting among these configuration options.  The
interdependencies between options are typically caused by interaction
in the resulting software system.  Constructing a well-functioning
software variant can be a complex and error-prone
process~\cite{DBLP:books/daglib/0032924}.

\emph{Feature models}~\cite{Batory:2005} allow us to organize the
configuration space and facilitate the construction of software
variants by describing configuration options using interdependent
\emph{features}~\cite{KangFeatureOrientedDomain1990}: a feature is a
name representing some functionality, a set of features is called a
\emph{configuration}, and each software variant is identified by a
\emph{valid configuration} (called a \emph{product}, for short).

Highly-configurable software systems can consist of thousands of
features and combine several subsystems
\cite{berger10ase,berger13vamos,tartler11eurosys,lienhardtDDP18}, each
with different features.  The construction and maintenance of feature
models with thousands of features for such highly-configurable
systems, can be simplified by representing large feature models as
sets of smaller interdependent feature
models~\cite{berger13vamos,Rosenmuller:2011:MVM} which we call
\emph{fragments}.  However, the analysis of such fragmented feature
models usually requires the fragments to be composed, to enable the
application of existing analysis
techniques~\cite{Benavides07fama:tooling,Mendonca:2009:SSP,
  DBLP:journals/is/BenavidesSC10,Segura:2012:BBT,THUM201470,
  Thum-EtAl:ACM-CS-2014}.  To this aim, many approaches for composing
feature models from fragments have been
investigated~\cite{10.1007/978-3-642-13595-8_3,10.1007/978-3-642-41533-3_22,10.1007/978-3-642-21210-9_11,CLASSEN20111130,Rosenmuller-et-al:McGPLE-2008,Schroter:2013:AAD}.

The analysis of fragmented feature models can be simplified if
suitable abstractions can safely replace some of the feature model
fragments in the analysis.  This simplification can be realized by
means of \emph{feature-model interfaces}~\cite{SchroterICSE2016}.  A
feature-model interface is a feature model that hides some of the
features and dependencies of another feature model (thus, interfaces
are closely related to \emph{feature-model
  slicing}~\cite{DBLP:conf/kbse/AcherCLF11}).  An interface can be
used instead of a feature model fragment to simplify the overall
feature model.  For certain analyses, working on the simplified
feature model produces results that also hold for the original feature
model and for any feature model where the interface is replaced by a
fragment compatible with the interface.

This paper addresses automated \emph{product discovery} in large
configuration spaces represented as sets of interdependent feature
models. Product discovery (sometimes called product configuration) is
a particular analysis for finding a product which includes a desired
set of features~\cite{DBLP:journals/computing/GalindoBTGR19}.  We aim
at automatically discovering a product that contains a given set of
features from the feature model fragments, without having to compose
all the fragments to apply the analysis.  This work is motivated by
our recent experiences in applying techniques for variability modeling
to automated product discovery in industrial use cases such as
Gentoo~\cite{gentoo}, a source-based Linux distribution that consists
of many highly-configurable packages.  The March 1st 2019 version of
the Gentoo distribution comprises 671617 features spread across 36197
feature models.  Gentoo's huge configuration space can be seen as the
composition of the feature models for all its packages, where package
interdependencies are modeled as shared features.  Gentoo's official
package manager and distribution system
\textsf{Portage}~\cite{Portage} achieves (via its \textsf{emerge}
tool) efficiency at the expense of completeness; i.e., in some cases
this tool fails to discover a product that contains a given set of
features, although such a product exists.  We show that feature model
interfaces~\cite{SchroterICSE2016}, which were developed to support
analysis reuse for feature model evolution in fragmented feature
models, do not allow us to reach our aim of complete and efficient
automated product discovery.

We propose a novel method for product discovery in sets of
interdependent feature models.  The proposed method is lazy in the
sense that features are added incrementally to the analysis until a
product is found.  We provide a formal account of the method and
evaluate it by implementing an efficient and complete dependency
solver for Gentoo. In short, our contributions are:

\begin{enumerate}
\item we strengthen feature model interfaces to enable lazy product
  discovery in sets of interdependent feature models;
\item we propose an efficient and complete algorithm for lazy product
  discovery in huge configuration spaces;
\item we provide an open-source implementation of the proposed
  algorithm;\footnote{The lazy product-discovery tool is available at
    \href{https://github.com/gzoumix/pdepa}{
      https://github.com/gzoumix/pdepa} and at
    \href{https://archive.softwareheritage.org/browse/origin/https://github.com/gzoumix/pdepa.git}{archive.softwareheritage.org/browse/origin/https://github.com/gzoumix/pdepa.git}}
  and
\item we evaluate the potential of lazy product discovery in terms of
  experiments on an industrial-size configuration space.\footnote{The
    evaluation artifact is available at
    \url{https://doi.org/10.6084/m9.figshare.11728914.v4} and
    \url{https://doi.org/10.5281/zenodo.3633643}}
\end{enumerate}

\section{Motivation and Overall Concept}\label{sec:motivation}
A software system like the Gentoo distribution comprises 36197
configurable packages, as of its March 1st 2019 version.  The
configuration space of each package can be represented by a feature
model; the overall configuration space of Gentoo can then be
represented by a feature model that is the composition of the feature
models of the 36197 packages.  The resulting feature model has 671617
features, and thus a configuration space with up to $2^{671617}$
solutions.

Gentoo's official package manager \textsf{Portage} implements an
optimized, heuristics-based product-discovery algorithm to find
products in this configuration space. This algorithm is not complete;
i.e., it fails to solve some product-discovery problems that have
solutions.  To the best of our knowledge, existing complete
product-discovery approaches need to load the entire feature model to
find products. Consequently, they do not scale to product-discovery
problems of the size of Gentoo's configuration space.

In this paper we target product discovery in huge configuration
spaces, such as for Gentoo, that can be described by a feature model
represented as a set $S$ of feature models with shared features, where
loading the overall feature model (i.e., the whole set $S$) is too
expensive.  We propose \emph{lazy product discovery}, a
product-discovery method that loads the elements of $S$ incrementally,
until it finds a product of the overall feature model.  The method
relies on the notion of a \emph{\conextslice\ of a feature model
  $\elFM$ for a set of features $Y$}. This is a feature model $\elFM'$
whose products are products of $\elFM$ and include all the products of
$\elFM$ that contain a feature in $Y$.

The proposed algorithm, shown in Listing~\ref{lst:algo}, takes as
input a set $S$ of feature models with shared features and a set $c$
of features to be included in the discovered product.  After
initialization, the algorithm incrementally loads \conextslices until
a solution has been found.  Let $\mathcal{M}_0$ denote the composition
of the feature models in $S$.  The algorithm returns a (not
necessarily minimal) product of $\mathcal{M}_0$ which includes all
features in $c$, whenever such a product exists; otherwise, it returns
the special value \textsf{None}.  The algorithm relies on the
following three auxiliary functions:
\begin{enumerate}
\item $\texttt{pick\_cut}(\elFM,Y)$: a function that, given a
  feature model $\elFM$ and a set of features $Y$, returns a \conextslice\ of
  $\elFM$ for $Y$;
\item $\texttt{compose}(\{\elFM_1, \dots, \elFM_n \})$: a function
  that, given a set of feature models $\elFM_1, \dots, \elFM_n$,
  returns the composition of the feature models in the set; and
\item $\texttt{select}(\elFM,c)$: a function that, given a feature
  model $\elFM$ and a set of features $c$, returns a product of
  $\elFM$ containing all the features in $c$ if it exists, and
  \lstinline[language=Algo]|None| otherwise.
\end{enumerate}

\begin{figure}[t]
\begin{lstlisting}[frame=top,frame=bottom,framexleftmargin=17pt,
  framexrightmargin=0pt,
  framexbottommargin=5pt,
  framextopmargin=5pt,language=Algo,caption={Lazy product-discovery algorithm},label={lst:algo},numbers=left,xleftmargin=5.0ex]
input $S$: set of feature models  $\label{line:fm}$
input $c$: configuration $\label{line:conf}$
var $Y$ = $c$ $\label{line:examined}$
var $\elFM'$ = $\texttt{compose}(\{\texttt{pick\_cut}(\elFM,Y) | \elFM \in S \})$ $\label{line:fp-start}$
var $\text{solution}$ = $\text{select}(\elFM', c)$ $\label{line:lookup1}$
while (solution $\neq$ None $\wedge$ solution $\not\subseteq Y$): $\label{line:while}$
  $Y$ = $Y \cup \text{solution}$ $\label{line:update-exmined}$
  $\elFM'$ = $\texttt{compose}(\{\texttt{pick\_cut}(\elFM,Y) | \elFM \in S \})$ $\label{line:update-cut}$
  $\text{solution}$ = $\text{select}(\elFM', c)$ $\label{line:fp-end}$ $\label{line:update-res}$
return $\text{solution}$ $\label{line:solution}$
\end{lstlisting}
\end{figure}

Assuming that the auxiliary functions (1), (2) and (3) work, we have
that on Line~\ref{line:while} the following loop invariants hold:
\begin{description}
\item[Inv1:]  $c\subseteq Y$.
\item[Inv2:] \textsf{solution} is a product of $\mathcal{M}'$ which
  includes all features in $c$, whenever such a product exists;
  otherwise \textsf{solution}  is the special value \lstinline[language=Algo]|None|.
\item[Inv3:] if \textsf{solution} is a product of $\mathcal{M}'$ and
  $\textsf{solution} \subseteq Y$, then \textsf{solution} is also a
  product of $\mathcal{M}_0$.
\item[Inv4:] If $\mathcal{M}'$ has no product which includes all
  features in $c$, then neither does $\mathcal{M}_0$.
\end{description}
Checking that \textbf{Inv1} holds is straightforward: just observe
that on Line~\ref{line:examined} the variable $Y$ is initialized to
$c$ and that at each iteration of the \lstinline[language=Algo]|while|
loop new features are added to $Y$ on Line~\ref{line:update-exmined}.
Checking that \textbf{Inv2} holds is equally straightforward:
according to the description of the auxiliary functions (1), (2)
and (3), the invariant is established on Lines~\ref{line:fp-start}
and~\ref{line:lookup1} as well as on Lines~\ref{line:update-cut}
and~\ref{line:update-res}.  The fact that \textbf{Inv3} and
\textbf{Inv4} hold is shown in the proof of Theorem~\ref{thm:lpd} in
Section~\ref{sec:product}.  The algorithm terminates because at each
iteration of the \lstinline[language=Algo]|while| loop, the size of
the set $Y$ (which, by construction, only contains features from the
features models in $S$) increases.  When the algorithm terminates we
have that either \textsf{solution} = \lstinline[language=Algo]|None|
or \lstinline[language=Algo]|None| $\not=$ \textsf{solution}
$\subseteq Y$.  In the first case (by \textbf{Inv4}) we have that
$\mathcal{M}_0$ has no product that contains all the features in $c$,
while in the second case (by \textbf{Inv3}) we have that
\textsf{solution} is a product of $\mathcal{M}_0$ that contains all
the features in $c$.

The laziness of this algorithm stems from the fact that it does not
need to consider $\mathcal{M}_0$ at once.  Instead, the algorithm
starts by considering the composition of the \conextslices\ of the feature
models for $Y=c$ and then iterates by considering bigger and bigger
\conextslices\ until the candidate solution is contained in the set $Y$. When
this happens we know, for the properties of the \conextslice, that the found
solution is also a solution for $\mathcal{M}_0$.

The algorithm's efficiency in finding a product with the features in
$c$ (see Lines~\ref{line:fp-start}, \ref{line:lookup1} and
\ref{line:update-cut}, \ref{line:update-res} of
Listing~\ref{lst:algo}) compared to executing
$\texttt{select}(\elFM_0,c)$, depends on the degree to which the
feature models in $S$ are such that:
 \begin{itemize}
 \item[-] 
computing  
 $\texttt{pick\_cut}(\elFM,Y)$ is
   efficient,
 \item[-] the feature models $\elFM'$ are small compared to $\elFM_0$,
 \item[-]
   $\texttt{select}(\elFM',c)$ performs better than
   $\texttt{select}(\elFM_0,c)$, and
 \item[-] a small number of iterations of the
   \lstinline[language=Algo]|while|-loop is required.
\end{itemize}

For the Gentoo distribution, each feature model $\elFM_i$ in $S$ has a
distinguished feature $f_i$ such that the constraints expressed by
$\elFM_i$ are enabled only if $f_i$ is selected (see
Section~\ref{sec:ExpDesign}). This reflects that each $\elFM_i$
corresponds to a Gentoo package that is installed if and only if $f_i$
is selected. Therefore, the function $\texttt{pick\_cut}(\elFM,Y)$ can
be efficiently implemented by returning $\elFM_i$ if $f_i\in Y$, and
by returning a feature model that expresses no constraints (and can,
therefore, be ignored by the composition that builds $\elFM'$)
otherwise.

The rest of this paper is organized as follows:
Sections~\ref{sec:feature-model}--\ref{sec:abstration-composition}
provide a formal account of the lazy product-discovery method that
culminates in the proof that \textbf{Inv3} and \textbf{Inv4} hold,
Section~\ref{sec:practice} evaluates the performance of the lazy
product-discovery algorithm by means of experiments, and
Sections~\ref{sec:related-work} and~\ref{sec:conclusion} discuss
related work and conclude the paper, respectively.

\section{A Formalization of Feature Models}\label{sec:feature-model}

This section  presents a formalization of feature models (FM)
and related notions, including feature model interfaces and composition.

\subsection{Feature Model Representations}\label{sec:feature-model-representations}

Different representations of feature models are discussed, e.g.,  by
Batory~\cite{Batory:2005}.  In this paper, we will rely on the
\emph{propositional formula} representation of feature models. In this
representation, a feature model is given by a pair $(\elFMF,\phi)$
where:
\begin{itemize}
\item[-]
 $\elFMF$ is a set of features, and
 \item[-]
 $\phi$ is a
propositional formula where the variables 
$x$ are feature names:
$
\;\phi \; ::=\; x \;\vert\; \phi \wedge \phi \;\vert\; \phi \vee \phi \;\vert\; \phi \rightarrow \phi \;\vert\; \neg \phi.
$
\end{itemize}
A propositional formula $\phi$ over a set of features $\elFMF$
represents the feature models whose products are configurations
$\{x_1,...,x_n\}\subseteq\elFMF$ ($n\ge 0$) such that $\phi$ is
satisfied by assigning value true to the variables $x_i$
($1\le i\le n$) and false to all other variables.

\begin{example}[A propositional representation of glibc FM]\label{example-fm-pf}
  Gentoo packages can be configured by selecting features (called
  \emph{use flags} in Gentoo), which may trigger dependencies or
  conflicts between packages.  Version 2.29 of the {\em glibc}
  library, that contains the core functionalities of most Linux
  systems, is provided by the package {\sf sys-libs/glibc-2.29-r2}
  (abbreviated to {\sf glibc} in the sequel).  This package has many
  dependencies, including (as expressed in Gentoo's notation):
{\rm \begin{lstlisting}[language=EBUILD] 
doc? ( sys-apps/texinfo )
      vanilla?( !sys-libs/timezone-data )
\end{lstlisting}}
This dependency expresses that {\sf glibc} requires the {\sf texinfo} documentation generator (provided by any version of the {\sf sys-apps/texinfo} package) whenever the feature {\sf doc} is selected
 and if the feature {\sf vanilla} is selected, then {\sf glibc} conflicts with any version of the time zone database (as stated with the {\sf !sys-libs/timezone-data} constraint).
These {\em dependencies} and {\em conflicts} can be expressed by a
feature model $(\elFMF_{\text{\sf glibc}},\phi_{\text{\sf glibc}})$ where
{\small 
\begin{align*}
\elFMF_{\text{\sf glibc}} =  & \;\{\text{{\sf glibc}, {\sf txinfo}, {\sf tzdata}, {\sf glibc:doc}, {\sf glibc:v}}\}, \mbox{ and}\\
\phi_{\text{\sf glibc}} =  & \;\text{\sf glibc} \rightarrow ((\text{\sf glibc:doc} \rightarrow \text{\sf txinfo}) \land (\text{\sf glibc:v} \rightarrow (\neg \text{\sf tzdata})).
\end{align*}} 
\noindent
Here, the feature
 {\sf glibc} represents the {\sf glibc} package;
 {\sf txinfo} represents any {\sf sys-apps/texinfo} package;
 {\sf tzdata} represents any version of the  {\sf sys-libs/timezone-data} package;
 and {\sf glibc:doc} (resp. {\sf glibc:v}) represents the {\sf glibc}'s {\sf doc} (resp. {\sf vanilla}) use flag.
\end{example}

The propositional representation of feature models works well in
practice~\cite{Mendonca:2009:SAF,DBLP:journals/is/BenavidesSC10,Thum-EtAl:ACM-CS-2014}
and we shall use it for the evaluation of the proposed method  (in Section~\ref{sec:practice}).
In contrast, to simplify the proofs, we follow Schr{\"o}ter
\etal~\cite{SchroterICSE2016} in using an extensional representation
of feature models to present our theory.

\begin{definition}[Feature model, extensional representation]\label{def:fm}
A {\em Feature Model} $\elFM$ is a pair $(\elFMF,\elFMP)$ where
 $\elFMF$ is a set of features and $\elFMP\subseteq2^{\elFMF}$ a set of 
products.
\end{definition}

\begin{example}[An extensional representation of glibc FM]
\label{example-fm}
Let $2^{X}$ denote the powerset of $X$.  The feature model of
Example~\ref{example-fm-pf} can be given an extensional representation
$\elFM_{\text{\sf glibc}}=(\elFMF_{\text{\sf glibc}},\elFMP_{\text{\sf
    glibc}})$ where $\elFMF_{\text{\sf glibc}}$ is the same as in
Example~\ref{example-fm-pf} and
{\small \begin{align*}
\elFMP_{\text{\sf glibc}} =& \{\{ \text{{\sf glibc}} \}, \{\text{{\sf glibc}, {\sf txinfo}} \}, \{\text{{\sf glibc}, {\sf tzdata}} \}, \{\text{{\sf glibc}, {\sf txinfo}, {\sf tzdata}} \}\}\ \cup\\
        & \{\{ \text{{\sf glibc}, {\sf glibc:doc}, {\sf txinfo}} \}, \{ 
\text{{\sf glibc}, {\sf glibc:doc}, {\sf txinfo}, {\sf tzdata}} \}\}\ \cup\\
        & \{\{ \text{{\sf glibc}, {\sf glibc:v}} \}, \{ \text{{\sf glibc}, {\sf glibc:v}, {\sf txinfo}} \}\}\ \cup\\
        & \{\{ \text{{\sf glibc}, {\sf glibc:doc}, {\sf glibc:v}, {\sf txinfo}} \}\}\ \cup\\
        & 2^{\{\text{{\sf txinfo}, {\sf tzdata}, {\sf glibc:doc}, {\sf glibc:v}}\}}.
\end{align*}}
\noindent
In the description of $\elFMP_{\text{\sf glibc}}$, the first line
contains products with {\sf glibc} but none of its use flags are
selected, so {\sf texinfo} and {\sf tzdata} can be freely installed;
the second line contains products with the use flag {\sf doc} selected
in {\sf glibc}, so a package of {\sf sys-apps/texinfo} is always
required; the third line contains products with the use flag {\sf
  vanilla} selected in {\sf glibc}, so no package of {\sf
  sys-libs/timezone-data} is allowed; the forth line contains products
with both {\sf glibc}'s use flags selected, so {\sf sys-apps/texinfo}
is mandatory and {\sf sys-libs/timezone-data} forbidden; finally, the
fifth line represents products without {\sf glibc}, so all
combinations of other features are possible, including the empty set.
\end{example}

\begin{definition}[Empty FM, void FMs, and pre-products]\label{def:FM-empty}
The \emph{empty} feature model, denoted  $\elFM_{\emptyset} =
(\emptyset,\{\emptyset\})$,  has no features and has just the empty product 
$\emptyset$.
A \emph{void} feature model is a feature model that has no products,
i.e., it has the form $(\elFMF,\emptyset)$ for some $\elFMF$.
A \emph{pre-product} of a feature model $\elFM$ is a configuration $c$
that can be extended to a product of $\elFM$ (more formally,
$c \subseteq p$ for some product $p$ of $\elFM$).
\end{definition}

Based on the above definition of a pre-product, we identify two related search
problems.

\begin{definition}[Feature compatibility, product discovery]\label{def:FM-compatibility}
  Consider a feature model $\elFM$ and a set of features $c$ in
  $\elFM$. The \emph{feature-compatibility problem} for $c$ in $\elFM$
  is the problem of determining whether $c$ is a pre-product of
  $\elFM$ (i.e., whether the features in $c$ are compatible with the
  products in $\elFM$).
  The \emph{product-discovery problem} for $c$ in $\elFM$ is the
  problem of finding a product of $\elFM$ that extends $c$.
\end{definition}
Clearly, the feature-compatibility problem for $c$ in $\elFM$ has a positive answer if and only if 
the product-discovery problem for $c$  in $\elFM$  has a solution.

\subsection{Feature Model Interfaces}\label{sec:feature-model-interface}
\emph{Feature model interfaces} were defined by Schr{\"o}ter
\etal~\cite{SchroterICSE2016} as a binary relation $\preceq$,
expressing that a feature model $\elFM'$ is an interface of a feature
model $\elFM$ if $\elFM'$ ignores some features of $\elFM$.

\begin{definition}[FM interface relation]\label{def:interface} 
 A feature  
model  $\elFM'=(\elFMF' ,$ $\elFMP')$ is  an \emph{interface} of feature 
model $\elFM=(\elFMF,\elFMP)$, 
denoted as $\elFM'\preceq\elFM$, iff
 $\elFMF'\subseteq\elFMF$ and
 $\elFMP'=\{p\cap\elFMF' \sht p\in\elFMP\}$. 
\end{definition}
Note that, for all feature models $\elFM'=(\elFMF',$ $\elFMP')$ and
$\elFM$, if $\elFM'\preceq\elFM$ then (i) all products of $\elFM'$
are pre-products of $\elFM$ and (ii) $\elFM'$ is the only interface
of $\elFM$ which has exactly the features $\elFMF'$ (i.e., $\elFM'$ is
completely determined by $\elFMF'$).

\begin{example}[An interface for glibc FM]
\label{example-interface}
The feature model
{\small \begin{align*}
\elFMF =& \{ \text{{\sf glibc}, {\sf glibc:v}} \}\\
\elFMP =& \{ \emptyset, \{\text{{\sf glibc}}\}, \{\text{{\sf glibc}, {\sf glibc:v}}\}\} 
\end{align*}}
\noindent
is the interface of the feature model $\elFM_{\text{\sf glibc}}$ from
Example~\ref{example-fm} that is determined by the features {\sf
  glibc} and {\sf glibc:v}. 
\end{example}

The interface relation for feature models is a \emph{partial order}
(i.e., it is reflexive, transitive and anti-symmetric) and the empty
feature model $\elFM_{\emptyset}$ is an interface of every non-void
feature model $\elFM$.  Moreover, $\elFM$ is void if and only if
$(\emptyset,\emptyset)\preceq\elFM$.

The notion of a feature model interface is closely related to that of
a \emph{feature model slice}, which was defined by Acher
\etal~\cite{DBLP:conf/kbse/AcherCLF11} as a unary operator
$\fmslice{Y}$ restricting a feature model to a set $Y$ of features.
Given a feature model $\elFM$, $\fmslice{Y}(\elFM)$ is the feature
model obtained from $\elFM$ by removing the features not in $Y$.

\begin{definition}[FM slice operator]\label{def:sclice} 
The \emph{slice operator} $\fmslice{Y}$ on feature models, where $Y$ is a set 
of features, is defined by:
\[
\fmslice{Y}((\elFMF,\elFMP))=(\elFMF\cap Y,\{p\cap Y\sht p\in\elFMP\}). 
\]
\end{definition}

Note that, for every feature model $\elFM=(\elFMF,\elFMP)$ and set of
features $Y$, the feature model $\fmslice{Y}(\elFM)=(\elFMF',\elFMP')$
is the unique interface of $\elFM$ such that $\elFMF'=\elFMF\cap Y$.
Moreover, for  every interface $\elFM_1=(\elFMF_1,\elFMP_1)$ of 
$\elFM$ it holds that $\elFM_1=\fmslice{\elFMF_1}(\elFM)$.

\begin{example}[A slice of glibc FM] \label{example-slice} The feature
  model interface in Example~\ref{example-interface} can be obtained
  by applying $\fmslice{\{\text{{\sf glibc}, {\sf glibc:v}}\}}$ to the
  feature model $\elFM_{\text{\sf glibc}}$ of
  Example~\ref{example-fm}. 
\end{example}

\subsection{Feature Model Composition}\label{sec:feature-model-composition}

Highly-configurable software systems often consist of many
interdependent, configurable
packages~\cite{Lotufo:2010:ELK,gentoo,lienhardtDDP18}.  The
variability constraints of each of these packages can be represented
by a feature model.  Therefore, configuring two (or more packages) in
such a way that they can be installed together corresponds to
identifying a product in a suitable \emph{composition} of their associated
feature models. In the propositional representation of feature models, such
composition corresponds to logical conjunction; i.e., the composition
of two feature models $(\elFMF_1,\phi_1)$ and $(\elFMF_2,\phi_2)$ is
the feature model
\[
(\elFMF_1\cup\elFMF_2, \phi_1\wedge\phi_2).
\]
In the extensional representation of feature models, this form of
composition corresponds to the binary operator $\fmplus$ of
Schr{\"o}ter \etal~\cite{SchroterICSE2016}, which is similar to the
join operator from  relational algebra~\cite{Codd:1970:RMD}.

\begin{definition}[FM composition]\label{def:composition}
The composition of two feature models $\elFM_1=(\elFMF_1, \elFMP_1)$ and  
$\elFM_2=(\elFMF_2, \elFMP_2)$, 
denoted $\elFM_1\mathrel\fmplus\elFM_2$, is the feature model defined by:
$$\elFM_1\mathrel\fmplus\elFM_2 = (\elFMF_1 \cup \elFMF_2 , \{p \cup q \, \sht  
\, p \in \elFMP_1, q \in \elFMP_2 , p \cap \elFMF_2 = q \cap \elFMF_1 \}).$$
\end{definition}

The composition operator $\fmplus$ is associative and commutative,
with $\elFM_{\emptyset}$ as identity element (i.e.,
$\elFM\mathrel\fmplus\elFM_{\emptyset}=\elFM$). Composing a feature
model with a void feature model yields a void feature model:
$(\elFMF_1, \elFMP_1)\mathrel\fmplus(\elFMF_2, \emptyset)$ $=$
$(\elFMF_1\cup\elFMF_2, \emptyset)$.

\begin{example}[Composing glibc and gnome-shell FMs]
\label{example-composition}
Let us consider another important package of the Gentoo distribution:
{\em gnome-shell}, a core component of the Gnome Desktop environment.
Version 3.30.2 of {\em gnome-shell} is provided by the package {\sf
  gnome-base/gnome-shell-3.30.2-r2} (abbreviated to {\sf g-shell} in
the sequel), and its dependencies include the following statement:
{\rm \begin{lstlisting}[language=EBUILD]
  networkmanager?( sys-libs/timezone-data ).
\end{lstlisting}}
\noindent
This dependency expresses that {\sf g-shell} requires any version of
the time zone database when the feature {\sf networkmanager} is
selected.

The \emph{propositional representation} of this dependency can be
captured by the feature model $(\elFMF_{\text{\sf g-shell}},$ $\phi_{\text{\sf g-shell}})$, where
{\small \begin{align*}
\elFMF_{\text{\sf g-shell}} =&\; \{ \text{{\sf g-shell}, {\sf tzdata}, {\sf g-shell:nm}} \}, \mbox{ and}\\
\phi_{\text{\sf g-shell}} =&\; \text{\sf g-shell} \rightarrow (\text{\sf g-shell:nm} \rightarrow \text{\sf tzdata}).
\end{align*}}
\noindent 
The corresponding \emph{extensional representation} of this feature model is 
$\elFM_{\text{\sf g-shell}} = (\elFMF_{\text{\sf g-shell}},$ $\elFMP_{\text{\sf 
g-shell}})$, where:
{\small \begin{align*}
  \elFMP_{\text{\sf g-shell}} =& \{\{ \text{{\sf g-shell}} \}, \{ \text{{\sf g-shell}, {\sf tzdata}} \}\} \cup \\
         & \{\{ \text{{\sf g-shell}, {\sf tzdata}, {\sf g-shell:nm}} \}\} \cup\\
         & 2^{\{ \text{{\sf tzdata}, {\sf g-shell:nm}}\}}.
\end{align*}}
\noindent
Here,
 the first line contains products with {\sf g-shell} but none of its
 use flags are selected: {\sf tzdata} can be freely selected;
 the second line is the product where {\sf g-shell:nm} is also
 selected and {\sf tzdata} becomes mandatory;
 finally, the third line represents products without {\sf g-shell}.

The \emph{propositional representation} 
of the
composition is the feature model $(\elFMF_{\text{\sf
      full}},\phi_{\text{\sf full}})$, where
{\small \begin{align*}
\elFMF_{\text{\sf full}}
  =&\; \elFMF_{\text{\sf glibc}}\cup\elFMF_{\text{\sf g-shell}}\\
  =&\; \{ \text{{\sf glibc}, {\sf txinfo}, {\sf tzdata}, {\sf g-shell}, {\sf glibc:doc}, {\sf glibc:v}, {\sf g-shell:nm}} \}, \mbox{ and}\\
\phi_{\text{\sf full}}
  =&\; \phi_{\text{\sf glibc}}\land\phi_{\text{\sf g-shell}} \\ 
  =&\; (\text{\sf glibc} \rightarrow ((\text{\sf glibc:doc} \rightarrow 
\text{\sf txinfo}) \land (\text{\sf glibc:v} \rightarrow (\neg \text{\sf 
tzdata}))) \land\\
   &\; (\text{\sf g-shell} \rightarrow (\text{\sf g-shell:nm} \rightarrow 
\text{\sf tzdata})).
\end{align*}}
\noindent
The \emph{extensional representation} of the composition is the
feature model
$\elFM_{\text{\sf full}} = \elFM_{\text{\sf glibc}} \fmplus
\elFM_{\text{\sf g-shell}} = (\elFMF_{\text{\sf
    full}},\elFMP_{\text{\sf full}})$ where
{\small \begin{align*}
\elFMP_{\text{\sf full}}
  =\;& \elFMP_{\text{\sf glibc}}\cup \elFMP_{\text{\sf g-shell}}\cup 
2^{\{\text{{\sf txinfo}, {\sf tzdata},
       {\sf glibc:doc}, {\sf glibc:v}, {\sf g-shell:nm}}\}}\ \cup\\
     & \{\{ \text{{\sf glibc}, {\sf g-shell}} \} \cup p\sht p\in
       2^{\{\text{{\sf txinfo}, {\sf tzdata}}\}}\}\ \cup\\
     & \{\{ \text{{\sf glibc}, {\sf glibc:doc}, {\sf txinfo}, {\sf
       g-shell}} \} \cup p\sht p\in 2^{\{\text{{\sf tzdata}}\}}\}\ \cup\\
     & \{\{ \text{{\sf glibc}, {\sf glibc:v}, {\sf g-shell}} \}\cup
       p\sht p\in 2^{\{\text{{\sf txinfo}}\}}\}\ \cup\\
     & \{\{ \text{{\sf glibc, {\sf g-shell}}, {\sf g-shell:nm}, {\sf
       tzdata}} \}\cup p\sht p\in 2^{\{\text{{\sf txinfo}}\}}\}\ \cup\\
     & \{\{ \text{{\sf glibc}, {\sf glibc:doc}, {\sf glibc:v}, {\sf txinfo}, {\sf g-shell}} \}\}\ \cup\\
     & \{\{ \text{{\sf glibc}, {\sf glibc:doc}, {\sf txinfo}, {\sf g-shell}, {\sf g-shell:nm}, {\sf tzdata}} \}\}.
\end{align*}}
\noindent
Here, the first line contains the products where {\sf glibc} and {\sf
  g-shell} do not interact, i.e., either when they are not installed,
or only one of them is installed; the second line contains the
products where both {\sf glibc} and {\sf g-shell} are installed, but
without use flags selected, so all optional package can be freely
selected; the third line contains the products with the {\sf glibc}'s
use flag {\sf doc} selected, so {\sf sys-apps/texinfo} becomes
mandatory; the fourth line contains the products with the {\sf
  glibc}'s use flag {\sf vanilla} selected, so {\sf
  sys-libs/timezone-data} is forbidden; the fifth line contains the
products with the {\sf g-shell}'s use flag {\sf vanilla} network
manager, so {\sf sys-libs/timezone-data} is mandatory; the sixth line
contains the product with {\sf glibc}'s both use flags selected and
the seventh line contains the product with {\sf glibc}'s use flag {\sf
  doc} and {\sf g-shell}'s use flag {\sf networkmanager} are selected.
\end{example}

\section{Problem Statement}\label{sec:problemStatement}
Many case studies show that the size of feature models used to model
real configuration spaces can be challenging for both humans and
machines
\cite{Tartler:2012:CCA,berger13vamos,THUM201470,SchroterICSE2016},
including the feature model for the source-based Linux distribution
Gentoo~\cite{gentoo} mentioned above.  The state-of-the-art strategy
used to address this challenge is to represent large feature models
by sets of  smaller interdependent feature
models~\cite{berger13vamos,Rosenmuller:2011:MVM}. The resulting
interdependencies between different feature models can be
expressed using shared features~\cite{SchroterICSE2016}.

The feature compatibility problem for a given set of features (see
Definition~\ref{def:FM-compatibility}) can be decided without first
composing the considered feature models when the feature models are
disjoint, as it suffices to inspect each feature model independently.
Namely, feature-model slices can be used to formulate a
\emph{feature-compatibility criterion} for the case with no shared
features between the feature models, as shown by the following
theorem:

\begin{theorem}[Feature-compatibility criterion for disjoint FMs]\label{the:Feature-Compatibilty-Without-Shared-Features}
Consider the feature models $\elFM_i=(\elFMF_i,\elFMP_i)$ $(1\le i\le n)$ with pairwise no shared features (i.e., $1\le i\not=j\le n$ implies $\elFMF_i\cap\elFMF_j=\emptyset$).
Then a configuration $c$ is a pre-product of  the feature model $\elFM = \bigfmplus_{1\le i\le n}\elFM_i$ if and only if
 $c$ is a subset of $\bigcup_{1\le i\le n}\elFMF_i$ and for all $\elFM_i$ the configuration $c\cap\elFMF_i$ is a product of $\,\fmslice{c}(\elFM_i)$.
\end{theorem}
\begin{proof}
Let $\elFM=(\elFMF,\elFMP)$.

{\bf Case $\Rightarrow$.}
Since $c$ is a pre-product of $\elFM$, by definition there exist $p\in\elFMP$ such that $c\subseteq p$.
Hence $c\subseteq\elFMF=\bigcup_{1\le i\le n}\elFMF_i$.
Let now consider $\fmslice{c}(\elFM_i)$ for any $1\le i\le n$: by definition $p\cap c\cap\elFMF_i$ is a product of this feature model, and by construction, $p\cap c\cap\elFMF_i=c\cap\elFMF_i$.
Hence, $c\cap\elFMF_i$ is a product of $\fmslice{c}(\elFM_i)$ for any $1\le i\le n$.

{\bf Case $\Leftarrow$.}
Since for any $1\le i\le n$, $c\cap\elFMF_i$ is a product of $\fmslice{c}(\elFM_i)$, there exist $p_i\in\elFMP_i$ such that $c\cap\elFMF_i=p_i\cap c$.
Let consider the configuration $p=\bigcup_{1\le i\le n} p_i$.
Since the feature models $\elFM_i$ do not share features, we have $p_i\cap\elFMF_j=\emptyset=p_j\cap\elFMF_i$ for all $1\le i\neq j\le n$.
Hence $p$ is a product of $\elFM$.
Moreover, we have that:
 $$p\cap c=\bigcup_{1\le i\le n} (p_i\cap c)=\bigcup_{1\le i\le n} (c\cap\elFMF_i)=c\cap\bigcup_{1\le i\le n} \elFMF_i=c.$$
Hence $c\subseteq p$ holds, which means that $c$ is a pre-product of $\elFM$.
\end{proof}

Unfortunately, the feature compatibility criterion of
Theorem~\ref{the:Feature-Compatibilty-Without-Shared-Features} does
not work for feature models with shared features. The problem can be
illustrated by the following example.

\begin{example}[Feature compatibility with shared features]\label{ex:failure-of-criterion}
  Consider the two feature models $\elFM_{\text{\sf glibc}}$ and
  $\elFM_{\text{\sf g-shell}}$ from Examples~\ref{example-fm}
  and~\ref{example-composition}, and the configuration
$c=\{\text{{\sf glibc}, {\sf glibc:v}, {\sf g-shell}, {\sf g-shell:nm}}\}.$
  We have
{\small \begin{align*}
  \fmslice{c}(\elFM_{\text{\sf glibc}})=&\;(\{\text{{\sf glibc}, {\sf glibc:v}}\}, 2^{\{\text{{\sf glibc}, {\sf glibc:v}}\}}), \mbox{ and}\\
  \fmslice{c}(\elFM_{\text{\sf g-shell}})=&\;(\{\text{{\sf g-shell}, {\sf g-shell:nm}}\}, 2^{\{\text{{\sf g-shell}, {\sf g-shell:nm}}\}}).
\end{align*}}
\noindent
Here, we have that
$c\subseteq\elFMF_{\text{\sf glibc}}\cup\elFMF_{\text{\sf g-shell}}$
and it is clear from the previous equation that
$c\cap\elFMF_{\text{\sf glibc}}=\{\text{{\sf glibc}, {\sf glibc:v}}\}$
is a product of $\fmslice{c}(\elFM_{\text{\sf glibc}})$ and that
$c\cap\elFMF_{\text{\sf g-shell}}=\{\text{{\sf g-shell}, {\sf
    g-shell:nm}}\}$
is a product of $\fmslice{c}(\elFM_{\text{\sf g-shell}})$.  However,
$c$ is not a pre-product of
$\elFM_{\text{\sf glibc}}\mathrel\fmplus\elFM_{\text{\sf g-shell}}$,
since the use flag {\sf g-shell:nm} requires a timezone database to be
installed while the use flag {\sf glibc:v} forbids it.
\end{example} 

In this paper we address complete and efficient product discovery in
sets of interdependent feature models.  To this aim, we define a novel
criterion which, given some selected features, enables solving the
product-discovery problem for a set of feature model fragments with
shared features, without composing all the fragments.

\section{Lazy Product Discovery}\label{sec:abstration-composition}\label{sec:product}
We are looking for a product-discovery criterion which works for
interdependent feature models, similar to how the
feature-compatibility criterion given in
Theorem~\ref{the:Feature-Compatibilty-Without-Shared-Features} works
for disjoint feature models.  The solution lies in a novel criterion
based on strengthening the feature model interfaces.  Given feature
models with shared features $\elFM_i=(\elFMF_i,\elFMP_i)$ and a set of
selected features $c$, we need feature model interfaces $\elFM_i'$
that reflect how $c$ is related to other features in $\elFM_i$ in
order to guarantee that the interface behaves similarly to $\elFM_i$
with respect to the feature-compatibility problem for $c$.  More
formally, the interface $\elFM'_i$ must satisfy the following
conditions:
\begin{enumerate}
\item $\fmslice{c}(\elFM_i)\preceq\elFM'_i\;$; and  
\item the products of $\elFM_i'$ are among the products of $\elFM_i$.
\end{enumerate}

\begin{example}[Feature compatibility with shared features continued]\label{ex:interraction}
  Consider feature models $\elFM_{\text{\sf glibc}}$ and
  $\elFM_{\text{\sf g-shell}}$ and configuration $c$, as discussed in
  Example~\ref{ex:failure-of-criterion}.  Let
  $c_1=\{\text{{\sf glibc}, {\sf glibc:v}}\}$ and
  $c_2=\{\text{{\sf glibc}, {\sf tzdata}, {\sf glibc:v}}\}$.  We can
  see that the interface
  $\elFM'_{\text{\sf glibc}} = \fmslice{c_2}(\elFM_{\text{\sf
      glibc}})$ of $\elFM_{\text{\sf glibc}}$ satisfies (with
  $i=\text{\sf glibc}$) conditions (1) and (2) above.  Since
  $\fmslice{c}(\elFM_{\text{\sf
      glibc}})=\fmslice{c_1}(\elFM_{\text{\sf glibc}})$ and
  $c_2\setminus c_1=\{\text{\sf tzdata}\}$, this shows that it is
  important to consider the feature {\sf tzdata} when checking whether
  $c$ is a pre-product of a composed feature model including
  $\elFM_{\text{\sf glibc}}$.
\end{example}

Let us now introduce terminology for different restrictions to the
interface relation that satisfy one or both of the conditions (1) and
(2) given above, and investigate some of their properties.

\begin{definition}[FM \extslice, conservative interface, and \conextslice\ relations]\label{def:conservative-interface}
Given a set of features $Y$ and two feature models $\elFM'=(\elFMF',\elFMP')$ and $\elFM=(\elFMF,\elFMP)$, 
 we say that 
\begin{enumerate}
\item\label{def:cut}
$\elFM'$ is an \emph{\extslice\ for $Y$} of $\elFM$, denoted $\elFM'\mathrel\preceq_Y\elFM$, iff
 $\fmslice{Y}(\elFM)\preceq\elFM'\preceq\elFM$ holds;
\item\label{def:c-i}
$\elFM'$ is a \emph{conservative interface} of $\elFM$, denoted $\elFM'\mathrel\fmcut\elFM$, 
iff    both $\elFM'\mathrel\preceq\elFM$ and 
$\elFMP'\subseteq\elFMP$ hold;  and
\item\label{def:c-cut} $\elFM'$ is a \emph{\conextslice\ for $Y$} of
  $\elFM$, denoted $\elFM'\mathrel\fmcut_Y\elFM$, iff $\elFM'$ is both
  an \extslice for $Y$ and a conservative interface.
\end{enumerate}
\end{definition}
Note that $\fmcut_\emptyset=\fmcut$. The relation $\fmcut$ is a
partial order; the feature model $(\emptyset,\emptyset)$ is the
\emph{minimum} (i.e., the smallest w.r.t.\ both $\preceq$ and
$\fmcut$) conservative interface of every void feature model; and the
empty feature model $\elFM_{\emptyset}$ is the minimum conservative
interface of every feature model that has the empty product.

The following theorem proves, in a constructive way, the existence of
the minimum \conextslice\ of $\elFM$ for $Y$, for any feature model
$\elFM=(\elFMF,\elFMP)$.  Let the \emph{minimal products} of $\elFM$
be the products that are not included in other products, and let
$Y'=(\elFMF\cap Y)$ be the set of features of $\elFM$ that occur in
$Y$.  Intuitively, the \emph{minimum \conextslice} of $\elFM$ for $Y$
is the feature model obtained from $(Y', \emptyset)$ by incrementally
adding all the minimal products of $\elFM$ (and their features) that
contain a feature occurring in the feature model, until a fixed point
is reached.

\begin{theorem}[Characterization of the minimum \conextslice]\label{th:characterization-minimum-cut}
For all sets $Y$ of features  and all feature models  $\elFM=(\elFMF, \elFMP)$, 
let $\fmcutminY(\elFM)$ be the minimum \conextslice\ of $\elFM$ for $Y$, i.e.,
$$\fmcutminY(\elFM) = \min\nolimits_{\fmcut}\{\elFM'\sht\elFM'\mathrel\fmcut_Y\elFM\}.$$
Then $\fmcutminY(\elFM)=f^\infty(((\elFMF\cap Y), \emptyset))$, where
$f$ is the function between feature models defined by
$$f((\elFMF_1,\elFMP_1))=(\elFMF_1\cup(\bigcup\limits_{p\in\elFMP_2}p),\elFMP_1\cup\elFMP_2)$$
with
$\elFMP_2=\{p\in\elFMP\sht \forall p'\in\elFMP,\,(p'\subsetneq
p)\Rightarrow((p\setminus p')\cap\elFMF_1\neq\emptyset)\}$.
\end{theorem}
\begin{proof}
Let $\elFM'=((\elFMF\cap Y), \emptyset)$ and consider the partially ordered set of feature models $(S,\leq)$, defined by
\begin{itemize}
\item[-]
$S = \{(\elFMF'',\elFMP'')\sht (\elFMF\cap Y)\subseteq\elFMF''\subseteq\elFMF \land \elFMP''\subseteq\elFMP\}$, and
\item[-]
$(\elFMF_1,\elFMP_1)\leq(\elFMF_2,\elFMP_2) \quad \text{iff}\quad  (\elFMF_1\subseteq\elFMF_2) \ \text{and}\ (\elFMP_1\subseteq\elFMP_2)$.
\end{itemize}
It is straightforward to see that $(S,\leq)$ is a complete lattice (with minimum $\elFM'$ and maximum $\elFM$) and that $f$ is monotonic increasing for $\leq$.
Hence, by~\cite{kleene:1938}, $f^\infty(\elFM')$ exists and is the minimum fixpoint of $f$.

We prove that the fixpoints of $f$ are exactly the \conextslices\ of
$\elFM$ for $Y$.  Let us first consider a feature model
$\elFM_Y=(\elFMF_Y,\elFMP_Y)$ that is a \conextslice\ of $\elFM$ for
$Y$.  Since $\elFM_Y=\fmslice{\elFMF_Y}(\elFM)$ and
$\elFMP_Y\subseteq\elFMP$ for all $p\in\elFMP$, we have
$p\cap\elFMF_Y\in\elFMP$.  This implies that for any
$p\in\elFMP\setminus\elFMP_Y$, there exists $p'\in\elFMP$ with
$p'\subsetneq p$ such that $(p\setminus p')\cap\elFMF_Y=\emptyset$.
By definition, we have
$f(\elFM_Y)=(\bigcup_{p\in\elFMP_2}p\cup\elFMF_Y,\elFMP_Y\cup\elFMP_2)$
with
$$\elFMP_2 = \{p\in\elFMP\sht\forall p'\in\elFMP,\,(p'\subsetneq p)\Rightarrow((p\setminus p')\cap\elFMF_Y\neq\emptyset)\} \subseteq \elFMP_Y.$$
Hence $f(\elFM_Y)=\elFM_Y$.

Let us now consider a feature model $\elFM_Y'=(\elFMF_Y',\elFMP_Y')$
in $S$ such that $f(\elFM_Y')=\elFM_Y'$.  First, it is clear by
construction that $\elFMP_Y'\subseteq\elFMP$.  Moreover, if we write
$\elFMP'=\{p\in\elFMP\sht\forall p'\in\elFMP\setminus\{p\},
p'\nsubseteq p\}$, it is clear from the definition of $f$ that
$\elFMP'\subseteq \elFMP_Y'$.  Suppose that the set
$M=\{p\in\elFMP\sht p\cap\elFMF_Y'\not\in\elFMP_Y'\}$ is not empty and
consider $p_1$ a minimal element of $M$ w.r.t. $\subseteq$.  Since
$p_1\nsubseteq\elFMF_Y'$, by definition of $\elFMP'$, the set
$N=\{p'\in\elFMP_Y'\sht p'\subseteq p_1\}$ is not empty. Consider any
maximal element $p_2$ of $N$ w.r.t. $\subseteq$.  Since
$p_1\cap\elFMF_Y'\not\in\elFMP_Y'$, we have
$(p_1\setminus p_2)\cap\elFMF_Y'\neq\emptyset$, and so the condition
$\forall p'\in\elFMP,\,(p'\subsetneq p_1)\Rightarrow((p_1\setminus
p')\cap\elFMF_Y'\neq\emptyset)$ holds.  It follows that $\elFM_Y'$ is
not a fixpoint of $f$ (since applying $f$ to $\elFM_Y'$ would add the
product $p_1$), which contradicts the hypothesis.  Hence for all
$p\in\elFMP$, $p\cap\elFMF_Y'\in\elFMP_Y'$, this means that
$\elFM_Y'=\fmslice{\elFMF_Y'}(\elFM)$.  Since by construction
$Y\cap\elFMF\subseteq\elFMF_Y'$, we have
$\fmslice{Y}(\elFM)\preceq\elFM_Y'\preceq\elFM$: $\elFM_Y'$ is a
\conextslice\ of $\elFM$ for $Y$.

To conclude, observe that the orders $\preceq$ and $\leq$ are equal on
the set of \conextslices\ of $\elFM$ for $Y$.  Since $f(\elFM')$ is
the minimum fixpoint of $f$ w.r.t. $\leq$, it is also the minimum
\conextslice\ of $\elFM$ for $Y$.
\end{proof}

\begin{example}[A minimum \conextslice\ \doneByFD{of } \text{\sf glibc} FM]
\label{ex:minimal_concut}
Consider the feature model $\elFM_{\text{\sf glibc}}$ of
Example~\ref{example-fm} and
$Y=\{\text{{\sf glibc}, {\sf glibc:doc}}\}$.  The minimal \conextslice\
$\fmcutminY(\elFM_{\text{\sf glibc}})$ can be computed by starting
with the feature model $(Y, \emptyset)$ and then applying $f$. In the
first application of $f$, the set $\elFMP_2$ collects the products
$\emptyset$, $\{ \text{{\sf glibc}} \}$,
$\{ \text{{\sf glibc:doc}} \}$, and
$\{ \text{{\sf glibc},{\sf glibc:doc},{\sf txinfo}} \}$. The set
$\elFMF_1$ after the first application becomes
$\{ \text{{\sf glibc},{\sf glibc:doc},{\sf txinfo}} \}$ and therefore,
in the second application of $f$, the products
$\{ \text{{\sf txinfo}} \}$, $\{ \text{{\sf glibc}, {\sf txinfo}}\}$,
and $\{ \text{{\sf glibc:doc}, {\sf txinfo}}\}$ are added to
$\elFMP_2$. At this point, further applications of $f$ do not add
further products.

In this case, the minimum \conextslice\
$\fmcutminY(\elFM_{\text{\sf glibc}})$ is different from the slice
$\fmslice{Y}(\elFM_{\text{\sf glibc}})$, since the \conextslice\ keeps the
information that when $\text{\sf glibc}$ and $\text{{\sf glibc:doc}}$
are selected, then $\text{{\sf txinfo}}$ also has to be selected.
\end{example}

The following theorem proves sufficient criteria to guarantee that a
product of the composition of \conextslices\ is also a product of the
composition of the original feature models and, conversely, that the
original feature model does not have a product that contains a given
set of features.  Intuitively, given a set of features $Y$ and a
product $p$ of the composition of \conextslices\ for $Y$, if $p$
is a subset of $Y$ we have that $p$ is also a product of the
composition of the original feature models.  Moreover, if the
composition of \conextslices\ for $Y$ has no products with the
features in a set $c\subseteq Y$, then neither does the the original
feature model.

\begin{theorem}[Product-discovery criterion for interdependent FMs]\label{th:valid-product}
  Consider a set $Y$ of features, a finite set $I$ of indices, and two
  sets of feature models $\{\elFM_i=(\elFMF_i,\elFMP_i)\sht i\in I\}$
  and $\{\elFM_i'=(\elFMF_i',\elFMP_i')\sht i\in I\}$ such that for
  all $i\in I$, $\elFM_i'\mathrel\fmcut_Y\elFM_i$.  Let
  $\elFM=(\elFMF,\elFMP)=\bigfmplus_{i\in I}\elFM_i$ and
  $\elFM'=(\elFMF',\elFMP')=\bigfmplus_{i\in I}\elFM_i'$. Then
\begin{enumerate}
\item
each product $p$ of $\elFM'$ such that $p\subseteq Y$  is a product of $\elFM$, and
\item for each set of features $c\subseteq Y$ and for each product $p$
  of $\elFM$ such that $c\subseteq p$, there exists a product $q$ of
  $\elFM'$ such that $c\subseteq q\subseteq p$.
\end{enumerate}
\end{theorem}
\begin{proof}
(1) Consider a product $p\in\elFMP'$. By construction, for every $i\in I$, there exists $p_i\in\elFMP_i'$ such that $p=\bigcup_{i\in I}p_i$ 
and, for all $i,j\in I$, $p_j\cap\elFMF_i'=p_i\cap\elFMF_i'$.
By Definition~\ref{def:conservative-interface}, for all $i\in I$, since $p_i\in\elFMP_i'$, we have that $p_i\in\elFMP_i$.
Let us now consider $i,j\in I$. We have that
$p_i\cap\elFMF_j = p_i\cap Y\cap\elFMF_j = p_i\cap\elFMF_j'=p_j\cap\elFMF_i' = p_j\cap Y\cap\elFMF_i = p_j\cap\elFMF_i$. 
Hence, $p=\bigcup_{i\in I}p_i\in\elFMP$.

(2)  By Definition~\ref{def:conservative-interface}, since
  $\elFM_i'\mathrel\preceq_Y\elFM_i$, we have
  $\fmslice{Y}(\elFM_i)\preceq\elFM_i'$. Then, for all $i\in I$, there
  exists $Y_i$ such that $c\subseteq Y\subseteq Y_i$ and
  $\fmslice{Y_i}(\elFM_i)=\elFM_i'$.  Consider a product $p\in\elFMP$
  such that $c\subseteq p$. By definition, for all $i\in I$, there
  exists $p_i\in\elFMP_i$ such that $p=\bigcup_{i\in I}p_i$ and for
  all $i,j\in I$, we have $p_i\cap\elFMF_j=p_j\cap\elFMF_i$.  Let
  $q=\bigcup_{i\in I}(p_i\cap Y_i)$. Clearly
  $c\subseteq q \subseteq p$. Moreover, consider $i,j\in I$; since
  $p_i\cap\elFMF_j=p_j\cap\elFMF_i$ holds, we have:
$(p_i\cap Y_i)\cap(\elFMF_j\cap Y_j) = (p_i\cap\elFMF_j)\cap(Y_i\cap Y_j) 
  = (p_j\cap\elFMF_i)\cap(Y_i\cap Y_j) = (p_j\cap Y_j)\cap(\elFMF_i\cap Y_i)$.
Hence $q\in\elFMP'$.
\end{proof}

\begin{example}[Using the product-discovery criterion with  $\text{\sf glibc}$ and $\text{\sf g-shell}$ FMs]
  Consider the packages $\text{\sf glibc}$ and $\text{\sf g-shell}$ of
  Example \ref{example-composition} and the
  set~$Y=\{\text{{\sf glibc}, {\sf glibc:v}, {\sf tzdata}}\}$. It is
  easy to see that the minimum \conextslice\ of
  $\elFM_{\text{\sf glibc}}$ for $Y$ is
  $\fmcutminY(\elFM_{\text{\sf glibc}}) = (Y, 2^Y \setminus Y)$
  because \text{\sf tzdata} can not be selected when \text{\sf glibc}
  and \text{\sf glibc:v} are selected.  Now consider the package
  $\text{\sf g-shell}$ instead. The minimum \conextslice\ of
  $\elFM_{\text{\sf g-shell}}$ for $Y$ is
  $\fmcutminY(\elFM_{\text{\sf g-shell}}) = (Y, 2^Y)$.
By the definition of feature model composition, we have that
$\fmcutminY(\elFM_{\text{\sf glibc}})
\bigfmplus \fmcutminY(\elFM_{\text{\sf 
g-shell}})$ is the same as
$\fmcutminY(\elFM_{\text{\sf glibc}})$.

Now, due to Theorem~\ref{th:valid-product}, we can for example derive
that the product $\{\text{{\sf glibc}, {\sf tzdata}}\}$ that contains
the shared feature \text{\sf tzdata} is also a product of the
composition of $\elFM_{\text{\sf glibc}}$ and
$\elFM_{\text{\sf g-shell}}$. Note that to discover this fact, we
avoided computing the composition of the entire feature models and
could ignore, e.g., features such as \text{\sf glibc:doc} and
\text{\sf g-shell}.
\end{example}

The criteria provided by Theorem~\ref{th:valid-product} allow us to
prove that the lazy product-discovery algorithm
(Listing~\ref{lst:algo} in Section~\ref{sec:motivation}) is correct
and complete.

\begin{theorem}[Soundness and completeness of lazy product
  discovery]\label{thm:lpd}
  Given a finite set $I$ of indices, a set of feature models
  $S=\{\elFM_i=(\elFMF_i,\elFMP_i)\sht i\in I\}$ such that all
  products of $\elFM_i$ are finite, and a finite configuration $c$,
  the lazy product-discovery algorithm (Listing~\ref{lst:algo})
  applied to $S$ and $c$ always finishes and returns a product of
  $\bigfmplus_{i\in I}\elFM_i$ that contains $c$ if and only if such a
  product exists.
\end{theorem}
\begin{proof}
  Recall the definitions of auxiliary functions
  (Section~\ref{sec:motivation}):
\begin{enumerate}
\item
$\texttt{pick\_cut}(\elFM,Y) =\ \elFM'$
  for some $\elFM'$ s.t.\ $\elFM' \mathrel\fmcut_Y \elFM$,
\item
  $\texttt{compose}(\{\elFM_1, \dots, \elFM_n \})$ $=$
  $\elFM_1\mathrel\fmplus\cdots\mathrel\fmplus\elFM_n$,
\item $\texttt{select}(\elFM,c)$ is a product of $\elFM$
  containing all the features in $c$ if such a product exists, \lstinline[language=Algo]|None|
  otherwise;
\end{enumerate}
and the loop invariants \textbf{Inv1}--\textbf{Inv4} on
Line~\ref{line:while}.  In Section~\ref{sec:motivation} we have 
already shown that the invariants \textbf{Inv1} and \textbf{Inv2}
hold, and that the algorithm always finishes (because the set of
examined features $Y$, which strictly increases during each traversal
of the \lstinline[language=Algo]|while| loop, is bounded by
$(\bigcup_{i\in I}\bigcup_{p\in\elFMP_i}p)\cup c$, which is finite by
hypothesis). We can now conclude the proof by observing that the
invariants \textbf{Inv3} and \textbf{Inv4} follow straightforwardly
from Theorem~\ref{th:valid-product}(1) and
Theorem~\ref{th:valid-product}(2), respectively.
\end{proof}

It is worth observing that a suitable structure of the feature models
can enable a particular efficient implementation of the function
$\texttt{pick\_cut}(\elFM,Y)$.  For instance, if the
feature-model $\elFM$ is propositionally represented with a pair of
the form $(\elFMF, f\rightarrow\psi)$ (for some set of features
$\elFMF$, feature $f\in\elFMF$ and formula $\psi$) then, whenever
$f\not\in Y$, $\texttt{pick\_cut}(\elFM,Y)$ can return the feature
model $(Y',2^{Y'})$ with $Y'=Y\cap\elFMF$, which corresponds to the
pair $(Y',\true)$ in propositional representation.  Therefore, feature
models of the form $(\elFMF, f\rightarrow\psi)$ such that $f\not\in Y$
can be filtered away before computing the composition
$\texttt{compose}(\{\texttt{pick\_cut}(\elFM,Y) | \elFM \in S\})$
in Lines~\ref{line:fp-start} and~\ref{line:update-cut} of the
algorithm.

\section{Evaluation}\label{sec:practice}

With lazy product discovery, we aim to efficiently address the
product-discovery problem in huge configuration spaces, consisting of
hundreds of thousands of features in tens of thousands of feature
models.  Therefore, we evaluate the performance of the lazy
product-discovery algorithm introduced in
Section~\ref{sec:motivation}.  The proposed algorithm loads feature
model fragments by need to examine specific features. A feature is
\emph{loaded} during a configuration process if it occurs in one of
the loaded feature model fragments.  In contrast, \emph{standard}
product-discovery algorithms
(e.g.,~\cite{THUM201470,hyvar-rec,Mendonca:2009:SSP}) load all the
feature models before the product-discovery process starts.

We compare the number of loaded features, the time, and the memory
needed to solve a product-discovery problem using a lazy and a
standard product-discovery algorithm.  In detail, we investigate the
following research questions:

\begin{RQ}\label{rq:gentoo:size}\em
  How is the number of loaded features affected by the choice of a
  lazy or a standard product-discovery algorithm?
\end{RQ}

\begin{RQ}\label{rq:gentoo:time}\em
  How are the speed and memory consumption of product discovery
  affected by the choice of a lazy or a standard product-discovery
  algorithm?
\end{RQ}

In industrial practice, product-discovery tools are often optimized
for efficiency at the expense of \emph{completeness}. As a
consequence, there may be product-discovery problems for which
solutions exist but no solution is found by the tool.  We compare the
lazy product-discovery algorithm to one such state-of-the-art tool by
looking at the percentage of cases in which no product is found by the
state-of-the-art tool (although products exists), and at the
difference in performance for cases when the state-of-the-art
product-discovery tool return a correct answer (that is, it either
discovers a product or fails when there are no products).  For this
purpose, we investigate the following research questions:

\begin{RQ}\label{rq:gentoo:completeness}\em
  How often does a state-of-the-art product-discovery tool fail
  because of its incompleteness (i.e., the tool does not discover any
  product, although there is at least one product)?
\end{RQ}

\begin{RQ}\label{rq:gentoo:time-vs-emerge}\em
  Is lazy product discovery a feasible alternative to state-of-the-art
  product-discovery tools in terms of execution time and memory
  consumption?
\end{RQ}

\subsection{Experimental Design and Subject}\label{sec:ExpDesign}

To answer these research questions, we performed experiments on an
industrial system with a huge configuration space. We chose Gentoo, a
source-based Linux distribution with highly-configurable
packages~\cite{gentoo}, which is among the largest fragmented
feature models studied in the literature~\cite{lienhardtDDP18}.  The
experiments were performed on the March 1st 2019 version of the
distribution, that contained 36197 feature models with 671617 features
overall.

There are no standard benchmarks for product reconfiguration requests.
Therefore, we constructed a set of 1000 product-discovery problems for
the evaluation.  The problems were generated by randomly selected a
set of features (between one and ten) such that each of these features
requires the installation of a different package.  Solving a
product-discovery problem $c$ in this context amounts to computing a
Gentoo product that includes any version of the packages associated to
the features in $c$ and of other packages such that that all
dependencies are fulfilled.

We implemented the algorithm of Listing~\ref{lst:algo} as a tool.
This tool, called \pdepa, targets Gentoo's package dependencies, which
are defined using an ad-hoc syntax~\cite{PMS}. As shown in
Example~\ref{example-fm-pf}, Gentoo's dependencies can be encoded into
feature models where features represent both packages and
configuration options (called \emph{use flags} in Gentoo).  \pdepa
parses a package dependency and generates the equivalent propositional
formula representing the package feature model. A particularity of
Gentoo is that the feature model of a package $f$ can be translated
into a propositional representation of the form
$(\elFMF, f\rightarrow\psi)$, where a package selection feature $f$
represents the package $f$.  The \pdepa tool exploits this structure
of the feature model in the implementation of the key functions
$\texttt{pick\_cut}$ and $\texttt{compose}$ by using the optimization
discussed at the end of Section~\ref{sec:product}.  Specifically,
\pdepa can avoid loading the feature models of packages whose package
selection feature is not in the set $Y$ of required features, when
composing \conextslices (Listing~\ref{lst:algo}, Lines 4 and 8).

As its solving engine, \pdepa uses the state-of-the-art SMT solver
Z3~\cite{z3}, known for its performance and expressivity.  Solvers
such as Z3 allow constraints to be added incrementally, reusing part
of the search done previously without always restarting the search
from scratch.  This is extremely useful for composing \conextslices\
(Listing~\ref{lst:algo}, Lines 4 and 8) since the existing constraints
can be reused, only adding incrementally the new constraints not
implied by the existing ones.  Although this does not formally reduce
the complexity of the algorithm, which is NP-hard in the worst
case,\footnote{The NP-hardness derives immediately from the
  NP-hardness of the problem of finding a valid model for a
  propositional formula.} in practice these optimizations enable a
significant speed-up.

To investigate the research question \emph{RQ~\ref{rq:gentoo:time}},
we need to compare \pdepa to a standard product-discovery algorithm.
Unfortunately, there is no off-the-shelf complete product-discovery
tool for Gentoo and therefore we implemented one to establish a
baseline for our experiments. We constructed a software that loads all
the feature models of all the Gentoo packages and then, as done by
\pdepa, calls the SMT solver Z3~\cite{z3} to solve the configuration
problem. We then compared the results of \pdepa to the corresponding
results of this \emph{baseline tool} (\emph{baseline} for short) in
terms of computation time and memory consumption. To ensure a fair
comparison, we employ a white-box evaluation, and both \pdepa and the
baseline use the same implementation for translating the Gentoo
dependencies and for loading the feature models.

For research questions \emph{RQ~\ref{rq:gentoo:completeness}} and
\emph{RQ~\ref{rq:gentoo:time-vs-emerge}}, we compare the results of
\pdepa to the corresponding results of optimized, heuristics-based
product-discovery with \textsf{emerge}, the command-line interface to
Gentoo's official package manager and distribution system
\textsf{Portage}, which is not complete (i.e., it fails to solve some
product-discovery problems that have solutions).

All experiments were performed on virtual machines provided by the
IaaS OpenStack cloud of the University of
Oslo.\footnote{\href{https://www.uio.no/english/services/it/hosting/iaas/}{https://www.uio.no/english/services/it/hosting/iaas/}}
Every virtual machine had 8 GB of RAM, 2 vCSPUs (2.5 GHz Intel Haswell
processors), and was running an Ubuntu 19.04 operating system. The
Gentoo operating system was virtualized by running Docker and the
image used for the experiments is publicly
available.\footnote{\href{https://hub.docker.com/r/gzoumix/pdepa}{https://hub.docker.com/r/gzoumix/pdepa}}

\subsection{Results and Discussion}\label{sec:ResultsAndDiscussions}
This section is organized according to research questions RQ1--RQ4.
To facilitate the discussion of the experiments, the figures
presenting the different results use a fixed ordering of the 1000
product-discovery problems we considered along the $x$-axis; this
ordering is determined by the number of features loaded by \pdepa
during its computation for a given problem. Each of the 1000
experiments was repeated 5 times for \pdepa, for emerge and for the
baseline; Figures~\ref{fig:rq1}--\ref{fig:rq2-mem} report
the mean values for each experiment.

\paragraph{\textbf{RQ~\ref{rq:gentoo:size}.}}
Figure~\ref{fig:rq1} shows the results of the experiments for research
question RQ~\ref{rq:gentoo:size} and reports on the number of features
loaded by \pdepa to solve each product-discovery problem.  To
highlight how lazy product discovery performs compared to standard
product discovery, which needs to load all features before the
analysis can start, these numbers are shown as the percentage of
features from the full feature model, for each of the product
discovery problems.  The product-discovery problems have been sorted
along the $x$-axis according to this percentage. The figure shows the
loaded features as a full line, the mean number for all the product
discovery problems as a dashed line, and the standard deviation
(abbreviated to SD in the figures) as a the bar. We see that for the
considered product-discovery problems, the mean number of loaded
features is only 1.53\% of the overall number of features. In summary,
the gain in loaded features when solving each of the considered 1000
product-discovery problems using lazy product discovery over standard
product discovery is significant.

\paragraph{\textbf{RQ~\ref{rq:gentoo:time}.}} For research question
RQ~\ref{rq:gentoo:time}, we compared the speed and memory consumption
of product discovery when using \pdepa and the baseline  on the
defined product-discovery problems.
For each problem, \pdepa loads parts of the FM and calls Z3
incrementally (until a valid product for the whole FM is found), while
the baseline first loads the whole FM and then calls Z3.

Figure~\ref{fig:rq4} shows the computation time for product discovery
using \pdepa (green line) and Figure~\ref{fig:rq2} shows the
computation time for product discovery using the baseline.  The mean
execution time for the baseline is 949 seconds, compared to 78 seconds
for \pdepa.  The minimum and maximum execution times of the baseline
are 861.9 and 1222.6 seconds, respectively.  The standard deviation
for the baseline is negligible (around 35 seconds).  It is worth
mentioning that about one third of the execution time is devoted to
loading the overall feature model, while the remaining time is taken
by Z3.  The minimum and maximum execution time of \pdepa are 1.7 and
155.22 seconds, respectively.  The standard deviation is lower than
the one for the baseline, about 18 seconds.  The maximum computation
time of \pdepa is less than one third of the computation time used by
the baseline to simply load the overall feature model, and it is about
the 16\% of the minimum execution time of the baseline.

Figure~\ref{fig:rq4-mem} shows the memory consumption for product
discovery using \pdepa (green line) and Figure~\ref{fig:rq2-mem} shows
the memory consumption for the baseline.  The mean memory consumption
for the baseline is 3,919.4 MB, compared to 400.715 MB for \pdepa.
The minimum and maximum memory consumption of the baseline are 3016
and 3980 MB, respectively. About 1 GB of the used memory here is for
the feature model itself.  The standard deviation for the baseline is
negligible (about 70.84 MB).  The 7 memory consumption values that
fall outside the standard deviation correspond to the product
discovery problems that have no solution.  The minimum and maximum
memory consumption of \pdepa are 73 and 620 MB, respectively.  The
standard deviation, 67.38 MB, is about the same as for the baseline.
The maximum memory consumption of \pdepa is about 19.62\% of the
minimum memory consumption of the baseline.

The experiments show a clear correlation between the time and the
memory taken by \pdepa to solve a product-discovery problem and the
number of features loaded by \pdepa (cf.\ Figure~\ref{fig:rq1}).

In summary, the experiments clearly demonstrate that lazy product
discovery allows significant speed-up and significant reduction
of memory consumption, compared to standard product discovery.

\xdefinecolor{darkgreen}{RGB}{175, 193, 36}

\pgfplotsset{
  legend style={
    legend cell align=left,
    nodes={scale=0.8, transform shape},
    column sep=.5em,
    anchor=north west
  },
  zone/.style={    /pgfplots/legend image code/.code={      \path[##1] (0cm,-0.1cm) rectangle (0.6cm,0.1cm);
    }
  },
  scaled ticks=false,
  yticklabel style={
    /pgf/number format/precision=3,
    /pgf/number format/fixed,
    /pgf/number format/fixed zerofill,
  },
  tickwidth=.1cm,
  tick align=outside,
  every tick label/.append style={font=\scriptsize}
}

\begin{figure}
\xdefinecolor{data}{RGB}{180, 130, 30}
\colorlet{mean}{black}
\def\min{0.020994}
\def\max{2.668336}
\def\mean{1.534100}
\def\derivation{0.303152}
\def\derivlow{1.230948}\def\derivhigh{1.837252}\begin{tikzpicture}
  \begin{axis}[
   legend style={
     at={(0.99,0.08)},
     legend columns=2,
     anchor=south east
   },
  width=.5\textwidth,
  height=4cm,
   xmin=0, xmax=1000,
   axis x line=bottom,
   xticklabel style={},
   ymin=0,
   y axis line style={},
   ytick={\min, \derivlow, \mean,   \derivhigh, \max},
  axis y line=left,
]

  \draw[opacity=0.2, fill=data, draw=data] (0,{\mean - \derivation}) rectangle (1000,{\mean + \derivation});
  \addplot[draw=data] table[x expr=\coordindex+1,y index=1,header=false] {data/pdepa_load.dat};
  \addlegendentry{loaded features}
  \draw[draw=mean, dashed] (0,\mean) -- (1000,\mean);
    \addlegendimage{zone,draw=none,fill=data,fill opacity=0.2}\addlegendentry{SD}
  \addlegendimage{no markers, mean, dashed}\addlegendentry{mean}
\end{axis}
\node[text width=10mm, anchor=east] at (0.42,2.7) {\emph{\footnotesize Percentage}};
\end{tikzpicture}
\vspace{-17.5pt}
\caption{Features loaded by \pdepa. \label{fig:rq1}}

\medskip

\begin{tikzpicture}
    \colorlet{pdepadata}{darkgreen!90}
  \colorlet{pdepamean}{black}
  \def\pdepamin{1.700000}
  \def\pdepamax{155.220000}
  \def\pdepamean{77.887170}
  \def\pdepaderivation{18.556784}
  \def\pdepaderivlow{59.330386}  \def\pdepaderivhigh{96.443954}  \colorlet{emergedata}{blue!90}
  \colorlet{emergemean}{red}
  \def\emergemin{0.810000}
  \def\emergemax{53.580000}
  \def\emergemean{6.897412}
  \def\emergederivation{5.671832}
  \def\emergederivlow{1.22558}  \def\emergederivhigh{12.569244}\begin{axis}[
  legend style={
    at={(0.02,1.05)},
    anchor=north west
  },
  extra y tick style={
    tick label style={yshift=-.3em}
  },
  width=.5\textwidth,
  height=6cm,
    xmin=0, xmax=1000,
  axis x line=bottom,
  xticklabel style={},
      ymin=0,
  y axis line style={},
  ytick={ \emergemean, \emergederivhigh, \emergemax, 
\pdepamin, \pdepaderivlow, \pdepamean, 
\pdepaderivhigh, \pdepamax},
  extra y ticks={\emergederivlow},
      axis y line=left,
  legend columns=3,
  transpose legend,
]

    \draw[opacity=0.2, fill=pdepadata, draw=pdepadata] (0,{\pdepamean - \pdepaderivation}) rectangle (1000,{\pdepamean + \pdepaderivation});
  \addplot[draw=pdepadata, thin] table[x expr=\coordindex+1,y index=1, header=false] {data/pdepa_time.dat}; \addlegendentry{\pdepa}
  \draw[draw=pdepamean, dashed] (0,\pdepamean) -- (1000,\pdepamean);
    \addlegendimage{no markers, pdepamean, dashed}\addlegendentry{\pdepa
    mean~~~~}
  \addlegendimage{zone,draw=none,fill=pdepadata,fill
    opacity=0.2}\addlegendentry{ \pdepa SD}

    \draw[opacity=0.2, fill=emergedata, draw=emergedata] (0,{\emergemean - \emergederivation}) rectangle (1000,{\emergemean + \emergederivation});
  \addplot[draw=emergedata, thin] table[x expr=\coordindex+1,y index=1, header=false] {data/emerge_time.dat}; \addlegendentry{emerge}
  \draw[draw=emergemean, dashed] (0,\emergemean) -- (1000,\emergemean);
    \addlegendimage{no markers, emergemean, dashed}\addlegendentry{emerge mean}
  \addlegendimage{zone,draw=none,fill=emergedata,fill opacity=0.2}\addlegendentry{emerge SD}
\end{axis}
\node[text width=10mm, anchor=east] at (0.22,4.7) {\emph{\footnotesize Seconds}};
\end{tikzpicture}
\vspace{-17.5pt}
\caption{Execution times for \pdepa
  and \textsf{emerge}.\label{fig:rq4}}

\medskip

\colorlet{data}{brown!90}
\colorlet{mean}{black}
\def\min{861.900000}
\def\max{1222.630000}
\def\mean{949.176420}
\def\derivation{34.947321}
\def\derivlow{914.229099}\def\derivhigh{984.123741}
\begin{tikzpicture}
  \begin{axis}[
      legend style={
     at={(0.99,0.05)},
     legend columns=3,
     anchor=south east
   },
  yticklabel style={
    /pgf/number format/precision=1,
    /pgf/number format/fixed,
    /pgf/number format/fixed zerofill,
  },
  extra y tick style={
    tick label style={yshift=-.3em}
  },
  width=.5\textwidth,
  height=4cm,
    xmin=0, xmax=1000,
  axis x line=bottom,
  xticklabel style={},
        ymin=840,
  y axis line style={},
  ytick={\derivlow, \mean,  \derivhigh, \max},
  extra y ticks={\min},
    axis y discontinuity=crunch,
  axis y line=left,
]

  \draw[opacity=0.2, fill=data, draw=data] (0,{\mean - \derivation}) rectangle (1000,{\mean + \derivation});
    \addplot[draw=data, thin] table[x expr=\coordindex+1,y index=1,header=false] {data/full_time.dat};
  \addlegendentry{baseline}
  \draw[draw=mean, dashed] (0,\mean) -- (1000,\mean);
    \addlegendimage{no markers, mean, dashed}\addlegendentry{mean}
  \addlegendimage{zone,draw=none,fill=data,fill opacity=0.2}\addlegendentry{SD}
\end{axis}
\node[text width=10mm, anchor=east] at (0.22,2.7) {\emph{\footnotesize Seconds}};
\end{tikzpicture}
\caption{Baseline execution time.\label{fig:rq2}}
 \end{figure}

 \begin{figure}
   \bigskip
\xdefinecolor{darkgreen}{RGB}{175, 193, 36}
\begin{tikzpicture}
    \colorlet{pdepadata}{darkgreen!90}
  \colorlet{pdepamean}{black}
  \def\pdepamin{73.000000}
  \def\pdepamax{620.000000}
  \def\pdepamean{400.715000}
  \def\pdepaderivation{67.384944}
  \def\pdepaderivlow{333.330056}  \def\pdepaderivhigh{468.099944}  \colorlet{emergedata}{blue!90}
  \colorlet{emergemean}{red}
  \def\emergemin{41.000000}
  \def\emergemax{186.000000}
  \def\emergemean{77.793000}
  \def\emergederivation{23.421765}
  \def\emergederivlow{54.371235}  \def\emergederivhigh{101.214765}\begin{axis}[
  legend style={
    at={(0.02,1.10)},
    anchor=north west
  },
  extra y tick style={
    tick label style={yshift=-.3em}
  },
  width=.5\textwidth,
  height=6cm,
    xmin=0, xmax=1000,
  axis x line=bottom,
  xticklabel style={},
      ymin=0,
  y axis line style={},
  ytick={ \emergemean, \emergederivhigh, \emergemax, 
 \pdepaderivlow, \pdepamean, 
\pdepaderivhigh, \pdepamax},
  extra y ticks={\emergederivlow},
      axis y line=left,
  legend columns=3,
  transpose legend,
]

    \draw[opacity=0.2, fill=pdepadata, draw=pdepadata] (0,{\pdepamean - \pdepaderivation}) rectangle (1000,{\pdepamean + \pdepaderivation});
  \addplot[draw=pdepadata, thin] table[x expr=\coordindex+1,y index=1, header=false] {data/pdepa_memory.dat}; \addlegendentry{\pdepa}
  \draw[draw=pdepamean, dashed] (0,\pdepamean) -- (1000,\pdepamean);
    \addlegendimage{no markers, pdepamean, dashed}\addlegendentry{\pdepa
    mean~~~~}
  \addlegendimage{zone,draw=none,fill=pdepadata,fill
    opacity=0.2}\addlegendentry{ \pdepa SD}

    \draw[opacity=0.2, fill=emergedata, draw=emergedata] (0,{\emergemean - \emergederivation}) rectangle (1000,{\emergemean + \emergederivation});
  \addplot[draw=emergedata, thin] table[x expr=\coordindex+1,y index=1, header=false] {data/emerge_memory.dat}; \addlegendentry{emerge}
  \draw[draw=emergemean, dashed] (0,\emergemean) -- (1000,\emergemean);
    \addlegendimage{no markers, emergemean, dashed}\addlegendentry{emerge mean}
  \addlegendimage{zone,draw=none,fill=emergedata,fill opacity=0.2}\addlegendentry{emerge SD}
\end{axis}
\node[text width=10mm, anchor=east] at (0.22,4.7) {\emph{\footnotesize MB}};
\end{tikzpicture}
\caption{Memory consumption for \pdepa and \textsf{emerge}.\label{fig:rq4-mem}}

\medskip

\colorlet{data}{brown!90}
\colorlet{mean}{black}
\def\min{3016.000000}
\def\max{3980.000000}
\def\mean{3919.431800}
\def\derivation{70.847219}
\def\derivlow{3848.584581}\def\derivhigh{3990.279019}
\begin{tikzpicture}
\begin{axis}[
  legend style={
    at={(.99,0.08)},
    anchor=south east
  },
  yticklabel style={
    /pgf/number format/precision=1,
    /pgf/number format/fixed,
    /pgf/number format/fixed zerofill,
  },
  extra y tick style={
    tick label style={yshift=-.3em}
  },
  width=.5\textwidth,
  height=4cm,
    xmin=0, xmax=1000,
  axis x line=bottom,
  xticklabel style={},
        ymin=3000,
  y axis line style={},
  ytick={\derivlow, \mean,  \derivhigh, \max},
  extra y ticks={\min},
    axis y discontinuity=crunch,
  axis y line=left,
]

  \draw[opacity=0.2, fill=data, draw=data] (0,{\mean - \derivation}) rectangle (1000,{\mean + \derivation});
    \addplot[draw=data, thin] table[x expr=\coordindex+1,y index=1, header=false] {data/full_memory.dat}; \addlegendentry{baseline}
  \draw[draw=mean, dashed] (0,\mean) -- (1000,\mean);
    \addlegendimage{no markers, mean, dashed}\addlegendentry{mean}
  \addlegendimage{zone,draw=none,fill=data,fill opacity=0.2}\addlegendentry{SD}
\end{axis}
\node[text width=10mm, anchor=east] at (0.22,2.7) {\emph{\footnotesize MB}};
\end{tikzpicture}
\caption{Baseline memory consumption.\label{fig:rq2-mem}}
 
\bigskip

\pgfplotsset{compat=1.11,
    /pgfplots/ybar legend/.style={
    /pgfplots/legend image code/.code={       \draw[##1,/tikz/.cd,yshift=-0.25em]
        (0cm,0cm) rectangle (3pt,0.8em);},
   },
}

\begin{tikzpicture}
  \xdefinecolor{commondata}{RGB}{128, 0, 84}
  \colorlet{emergedata}{blue!40}
\begin{axis}[
  legend style={
    at={(0.5,-1.5)},anchor=north,
    legend columns=2
  },
  reverse legend,
  yticklabel style={
    color=white,
  },
  width=.5\textwidth,
  height=2cm,
    xmin=0, xmax=1000,
  axis x line=bottom,
  xticklabel style={},
      ymin=0,
  ytick={0,1},
    axis y line=left,
  ybar=0pt,
  bar width=.00005\textwidth,
]

  \addplot[draw=emergedata, fill=emergedata] table[x index=1,y index=2, header=false] {data/emerge_fail.dat}; \addlegendentry{\textsf{emerge} failure}
  \addplot[draw=commondata, fill=commondata] table[x index=1,y index=2, header=false] {data/common_fail.dat};
  \addlegendentry{no solution exists}
\end{axis}
\node[text width=15mm, anchor=west] at (-0.9,0.6) {\emph{\footnotesize
No solution\\[-2pt] found}};
\end{tikzpicture}
\caption{Product-discovery problems with no solution and \textsf{emerge} failures.\label{fig:rq3}}
\end{figure}

\paragraph{\textbf{RQ~\ref{rq:gentoo:completeness}.}} We investigated
the failures of a heuristics-based incomplete product-discovery tool
(\textsf{emerge}) compared to the cases when the complete lazy product
discovery algorithm showed that no solution exists, for the 1000
considered product-discovery problems.  Figure~\ref{fig:rq3} shows the
product-discovery problems for which \textsf{emerge} does not find a
product (red and blue bars). For the considered product-discovery
problems, \textsf{emerge} fails to find a valid configuration in
26.7\% of the cases.  In 0,7\% of the cases (red bars), no solution
exists.  Therefore, in 26\% of the cases, \textsf{emerge} fails to
solve a product-discovery problem that has a solution.  The
experiments show an interesting correlation between the failures of
\textsf{emerge} observed in Figure~\ref{fig:rq3} and the number of
features loaded by \pdepa during the product-discovery process: the
failures of \textsf{emerge} occur more frequently as the number of
loaded features needed for lazy product discovery increases. This can
be seen since the sorting of the $x$-axis is the same in
Figures~\ref{fig:rq1} and~\ref{fig:rq3}. In summary, on 1000 randomly
selected product-discovery problems, \textsf{emerge} fails to find a
solution that exists in around 26\% of the cases.

\paragraph{\textbf{RQ~\ref{rq:gentoo:time-vs-emerge}.}}
For research question RQ~\ref{rq:gentoo:time-vs-emerge}, we
investigated how well \pdepa performs as an alternative to the
state-of-the-art configuration tool
\textsf{emerge}. Figure~\ref{fig:rq4} shows the time for product
discovery using \pdepa (green line) and \textsf{emerge} (blue line).
The light green and the light blue bars show the standard deviations
and the correspondingly colored dashed lines show the mean times in
seconds for \pdepa and \textsf{emerge}, respectively. The difference
in mean times suggests that \pdepa is 11.29 times slower than
\textsf{emerge} in average, which corresponds to 70 additional
seconds.  However, as the results for RQ~\ref{rq:gentoo:completeness}
above shows that \textsf{emerge} fails for a significant number of the
considered product-discovery problems, lazy product discovery appears
to be a feasible alternative to \textsf{emerge}.

Figure~\ref{fig:rq4-mem} shows the memory consumption for product
discovery using \pdepa (green line) and \textsf{emerge} (blue line).
The light green and the light blue bars show the standard deviations
and the corresponding colored dashed lines show the mean memory
consumption in MB for \pdepa and \textsf{emerge}, respectively. The
difference in mean times suggests that \pdepa consumes four times more
memory than \textsf{emerge} in average (which amounts to around 300
MB).

In summary, lazy product discovery appears as a feasible alternative
to \textsf{emerge} if around one order of magnitude additional
computation time and four times additional memory consumption are
acceptable to always find products when these exist.

\subsection{Threats to Validity}\label{sec:Threats}

\subsubsection{External Validity} 
The results of the evaluation strongly depend on the product-discovery
problems considered in the experiments, i.e., on the feature models of
the Gentoo packages identified by the features in each
product-discovery problem.  Due to the lack of standard benchmarks, we
considered 1000 product-discovery problems that were randomly selected
from the 671617 features of the March 1st 2019 version of the Gentoo
distribution.  The random selection used the standard {\tt random}
python library~\cite{random}, that allows to get a set of elements
uniformly chosen from a given set.

Different product-discovery problems could potentially lead to
different results. We plan to investigate other product-discovery
problems for Gentoo and for other domains to get more insights.  In
particular, it would be interesting to investigate how lazy product
discovery performs when varying both the size and the amount of
interdependencies of the feature models (see
Section~\ref{sec:motivation}).

\subsubsection{Internal Validity} 
We used prototype implementations of the lazy product-discovery
algorithm and of the standard product-discovery algorithm.  Both
implementations rely on the Z3 solver~\cite{z3}. Z3 was chosen because
it is a mature solver and freely available.  The standard
product-discovery algorithm just performs a call to the Z3 solver.
The lazy product-discovery algorithm calls the Z3 solver whenever a
new feature fragment is loaded.  Using a different solver than Z3 may
affect the execution time and memory consumption of both the standard
and the lazy product-discovery algorithms. We plan to repeat the
experiments using another solver.

Introducing optimizations in the lazy product-discovery algorithm
could potentially reduce the number of loaded features, the execution
time, and the memory consumption for the algorithm.  One possible
optimization could be to pre-compute at compile time the modal
implication graphs~\cite{Krieter:2018:PCD,DBLP:conf/sigsoft/CosmoV11}
of features, which could potentially avoid loading feature models
that, e.g., are found to be conflicting in the pre-analysis. Another
possible optimization could be the definition and usage of an ad-hoc
search strategy for the back-end solver, instead of using solver's
default search strategy.

Another threat to validity is that Gentoo's package dependencies are
not formally specified, but only given in a textual representation. To
reduce the probability of errors in the implementation of the lazy
product-discovery algorithm, we have used unit tests to compare the
results of \pdepa with known correct products.  These unit tests were
performed by extending the package repository of \textsf{portage} with
custom testing and interdependent packages.

Possible bugs in Gentoo's package manager may also be considered a
threat to validity.  When performing the experiments, we identified
the following surprising behavior in \textsf{emerge}:
\begin{enumerate}
\item For some sets of packages\footnote{These sets consisted of packages with an
    identical SLOT ~\cite{PMS}. SLOTs are used in \textsf{portage} to
    identify which versions of the same package can coexist in one
    system.}, \textsf{emerge} implements a heuristic that only
  considers the feature model of the most recent package in the set,
  thus forgetting possible solutions.
\item For \textsf{emerge} to consider a package, some part of its
  feature model 
  must be configured. Specifically, some of its features must be
  selected or deselected such that the constraint identified by the
  variable REQUIRED\_USE~\cite{PMS} evaluates to \emph{true}.
\item For a given product-discovery problem, the dependency analysis
  of \textsf{emerge} considers each package individually.  This can
  trigger the installation of a package in conflict with the rest of
  the product-discovery problem, thus preventing the product-discovery
  problem to be solved even if it has a solution.
\end{enumerate}
We reported these issues to the Gentoo developer community, which
replied that they could be considered as bugs of \textsf{emerge}.

We were not able to install the Gentoo variants corresponding to the
products discovered by \pdepa because of Bug (3) above.  Indeed, in
many cases, \textsf{emerge}'s dependency solver triggers the
installation of packages that conflict with \pdepa's solution.  We
plan to overcome this limitation by extending \pdepa into a complete
package installation tool for Gentoo.

\section{Related Work}\label{sec:related-work}
We discuss related work on interfaces, composition, and configuration
of feature models.

\subsubsection*{Interfaces of Feature Models}

The feature-model \conextslice\ in this paper strengthens the
feature-model interfaces introduced by Schr{\"o}ter
\etal~\cite{SchroterICSE2016}, which, as pointed out in
Section~\ref{sec:feature-model-interface}, are closely related to
\emph{feature model slices} introduced by Acher
\etal~\cite{DBLP:conf/kbse/AcherCLF11}.  In the work of Acher
\etal~\cite{DBLP:conf/kbse/AcherCLF11}, the focus is on feature model
decomposition. In subsequent work~\cite{Acher:2014:EEA}, Acher \etal
address evolutionary changes for extracted variability models by using
the slice operator in combination with a merge operator, and focus on
detecting differences between feature-model versions during evolution.
Instead, Schr{\"o}ter \etal~\cite{SchroterICSE2016} study how feature
model interfaces can be used to support evolution for a feature model
composed from feature models fragments. Changes to fragments which do
not affect their interfaces do not require the overall feature model
to be rebuilt (by composing the fragments) in order to reanalyze it.
Challenges encountered to support evolution in software product line
engineering have previously been studied by Dhungana
\etal~\cite{DHUNGANA20101108}.  They use interfaces to hide
information in feature model fragments and save a merge history of
fragments to give feedback and facilitate fragment maintenance. No
automated analysis is considered.  In contrast to this  work on
feature model interfaces for evolution, the \conextslice\ in our work is 
for efficient automated product discovery in huge feature models
represented as interdependent feature model fragments.

\emph{Feature-model
  views}~\cite{10.1007/978-3-642-14192-8_12,DBLP:conf/compsac/MannionSA09,10.1007/978-3-642-33666-9_17}
focus on a subset of the relevant features of a given feature model,
similarly to feature-model interfaces.  Different views regarding one
master feature model are used to capture the needs of different
stakeholders, so that a product of the master feature model can be
identified based on the views' partial configurations.  This work on
multiple views to a product in a feature model is orthogonal to our
work on feature-model \conextslices, which targets the efficient
configuration of systems comprising many interdependent configurable
packages.

\subsubsection*{Composition of Feature Models}

Feature-model composition is often used for multi software product
lines (i.e., sets of interdependent product
lines)~\cite{IST2012Holl,DBLP:conf/splc/Krueger06,DBLP:conf/vamos/RosenmullerS10,lienhardtDDP18}.
Eichelberger and Schmid~\cite{Eichelberger:2013:SAT} provide an
overview of textual-modeling languages which support variability-model
composition (like FAMILIAR~\cite{ACHER2013657},
VELVET~\cite{Rosenmuller:2011:MVM}, TVL~\cite{CLASSEN20111130},
VSL~\cite{DBLP:conf/vamos/AbelePSTW10}) and compare how they support
composition, modularity, and evolution.  Acher
\etal~\cite{10.1007/978-3-642-41533-3_22} compare different
feature-model composition operators by considering possible
implementations and discuss advantages and drawbacks.  For the
investigation of efficient automated configuration of huge feature
models in this paper, we use the propositional representation of
feature models and a composition operator that corresponds to logical
conjunction.

\subsubsection*{Configuration of Feature Models}

Product discovery (also called product configuration or product
derivation) is the process of selecting and deselecting features in a
feature model in order to obtain a
product~\cite{DBLP:journals/computing/GalindoBTGR19}.  This is a
central and widely studied problem in the field of automated reasoning
\cite{DBLP:journals/is/BenavidesSC10}; e.g., more than 50 different
methods for product discovery are discussed in a recent
survey~\cite{DBLP:journals/computing/GalindoBTGR19}.

We are not aware of any method that addresses how complete and
efficient product-discovery can be achieved in configuration spaces
comprising different interdependent feature model fragments without
composing all the fragments.  The tool for lazy product discovery
is in the class of product discovery tools which
\emph{automatically} produce valid configurations.

Automated configuration is supported by a number of tools, including
FeatureIDE~\cite{THUM201470}, GEARS~\cite{DBLP:conf/splc/KruegerC18a},
GUIDSL~\cite{Batory:2005}, IBED~\cite{DBLP:journals/asc/XueZT0CC016},
HyVarRec~\cite{hyvar-rec-csp},
SATIBEA~\cite{DBLP:conf/icse/HenardPHT15} S2T2
Configurator~\cite{DBLP:conf/vamos/BotterweckJS09},
SIP~\cite{DBLP:journals/tosem/HieronsLLSZ16}, SPL
Conqueror~\cite{DBLP:journals/sqj/SiegmundRKKAS12},
S.P.L.O.T.~\cite{Mendonca:2009:SSP}, and
VariaMos~\cite{DBLP:conf/caise/MazoSD12}.  However, in contrast to our
work, all these tools are eager and require the building of the global
feature model by composing all its fragments. As such, these tools are
in line with the standard product discovery algorithm, as discussed in
Section~\ref{sec:practice}.

Some of these standard product discovery tools are interactive, i.e.,
they support and interact with the user by guiding her in producing a
valid configuration or finding one that maximizes her preferences
\cite{Batory:2005,Mendonca:2009:SSP,DBLP:conf/vamos/BotterweckJS09,
  DBLP:conf/caise/MazoSD12}.  Our method for lazy product discovery
can be exploited to support interactive product discovery either (i)
by requiring the user to enter preferences over different
configurations or (ii) by interacting with the user when deciding what
partial configuration should be extended (i.e., when the
$\texttt{select}$ function of the algorithm in Listing~\ref{lst:algo}
is performed). An extension of the lazy product discovery algorithm in
this direction is left as future work.

Different computational techniques can be used to solve the product
discovery problem: satisfiability solvers, constraint programming,
evolutionary algorithms, stochastic algorithms, or binary decision
diagrams \cite{DBLP:conf/vamos/OchoaPRCS17,Benavides07fama:tooling,
  DBLP:journals/is/BenavidesSC10}.  Due to the NP-hardness of the
configuration problem itself, most complete approaches rely on SAT
solvers~\cite{DBLP:conf/splc/Janota08,Mendonca:2009:SAF}, but more
recently, the use of more powerful backend solvers, such as constraint
solvers and SMT solvers, are starting to be explored for automatic
configuration of feature models \cite{hyvar-rec,featureIDE-SMT,
  DBLP:conf/cade/MichelHGH12,DBLP:conf/seke/BenavidesTC05}.  In our
work, we have used Z3~\cite{z3} which is one of the most powerful and
mature SMT solvers available today. We would like to remark, however,
that the lazy product discovery method itself is orthogonal to the
tool chosen, as long as the backend solver allows to implement the
$\texttt{pick\_cut}$, $\texttt{compose}$, and $\texttt{select}$
operations of Listing~\ref{lst:algo}.

\section{Conclusion and Future Work}\label{sec:conclusion}
Product discovery in huge configuration spaces represented as sets of
interdependent feature models is challenging.  Standard analysis
techniques for fragmented feature models require all the feature
models to be composed in order to apply the analysis.  Recent work has
shown that several analyses of fragmented feature models can be
simplified using techniques such as feature model interfaces and
slicing, however these techniques do not work for product discovery in
sets of interdependent feature models.

In this paper, we introduce a method for automated product discovery
in configuration spaces represented as sets of interdependent feature
models.  The method is lazy as features are added incrementally to the
analysis until a product is found.  We introduce and formalize the
feature model \conextslice, and leverage this concept to define a
product-discovery criterion.  We exploit this criterion to define a
complete and efficient algorithm for lazy product discovery in sets of
interdependent feature models.  We have evaluated the potential of
lazy product discovery on randomly constructed configuration problems
for the configuration space of the source-based Linux distribution
Gentoo, with 36197 interdependent feature models and a total of 671617
features. The evaluation has demonstrated significant gains compared
to standard product discovery and that the trade-off of performance
for completeness is reasonable compared to the heuristics-based
product-discovery with \textsf{emerge}, the command-line interface to
Gentoo's official package manager and distribution system
\textsf{Portage}.
 
We are now investigating different optimizations of the current
prototype, such as the exploitation of modal implication graphs
pre-computed at compile time and the usage of ad-hoc SMT search
strategies.  In future work we plan to investigate other
product-discovery problems for Gentoo as well as for other domains, to
gain more insights into lazy product discovery.  While our results
make us confident that lazy product discovery is a viable method for
product discovery in huge configuration spaces, we believe that it may
also be used to complement optimized but incomplete algorithms when
these fail, such as \textsf{emerge} for Gentoo.  We also plan to
investigate how lazy product discovery can be combined with
interactive product discovery.

\begin{acks}
This work is partially funded by the Sirius Center for
    Scalable Data Access and the Compagnia di San Paolo. We thank
    the reviewers for constructive feedback, Thomas Th\"um, Andrzej
    Wasowski and Sven Apel for useful discussions on the topic of this
    paper, and Simone Donetti for testing the publicly available
    artifact.
\end{acks}




\end{document}